\documentclass[a4paper,11pt]{article}

\date{}
\usepackage{amsmath}
\usepackage{amstext}
\usepackage{psfrag}
\usepackage{amsthm}
\usepackage{amssymb}
\usepackage{tabularx}
\usepackage[latin2]{inputenc}
\usepackage{graphicx}
\usepackage{booktabs}

\usepackage{thm-restate}
\usepackage{thm-restate}

\newtheorem{theorem}{Theorem}[section]
\newtheorem{lemma}[theorem]{Lemma}
\newtheorem{claim}[theorem]{Claim}

\newtheorem{proposition}[theorem]{Proposition}
\newtheorem{observation}[theorem]{Observation}
\newtheorem{definition}[theorem]{Definition}

\usepackage{todonotes}

\newcommand{\longversion}[1]{}

\usepackage{epsfig}
\usepackage{tikz}

\usepackage{vmargin}
\setmarginsrb{1in}{1in}{1in}{1in}{0pt}{0pt}{0pt}{6mm}

\usepackage{pgf}
\usepackage{paralist} 
\usepackage{graphicx}
\usepackage{boxedminipage}
\usepackage{boxedminipage}
\usepackage{xcolor}
\usepackage{framed}
\usepackage{multicol}

\newcommand{\FPT}{\textsf{FPT}}

\usepackage{amsmath, amssymb, latexsym}
\usepackage{enumerate}

\usepackage{float}
\usepackage{xspace}

\newcommand{\untangle}{{\sf UT}}
\newcommand{\ep}{{\sf EP}}
\newcommand{\comp}{{\sf Comp}}
\newcommand{\per}{{\sf Per}}

\newcommand{\bigoh}{{\mathcal{O}}}

\usepackage{marvosym}

\newcommand{{\semiiso}}{{\sc SEMI-ISO}}

\newcommand{\erdosposa}{Erd\H{o}s-P\'{o}sa}

\newcommand{\cB}{\mathcal{B}}

\newcommand{\cP}{\mathcal{P}}

\newcommand{\cF}{\mathcal{F}}

\newcommand{\cW}{\mathcal{W}}
\newcommand{\cX}{\mathcal{X}}

\newcommand{\cZ}{\mathcal{Z}}

\newcommand{\bA}{\mathbf{A}}

\begin{document}
\title{The Half-integral {\erdosposa} Property for Non-null Cycles }
\author{
 Daniel Lokshtanov
 \thanks{Department of Informatics, University of Bergen \texttt{daniello@ii.uib.no}}\and 
  M. S. Ramanujan\thanks{Algorithms and Complexity Group, TU Wien \texttt{ramanujan@ac.tuwien.ac.at}}
  \and  
  Saket Saurabh  \thanks{The Institute of Mathematical Sciences\texttt{saket@imsc.res.in}}\addtocounter{footnote}{-3}\footnotemark
       }   
       
\maketitle
\thispagestyle{empty}

\begin{abstract} 
A Group Labeled Graph is a pair $(G,\Lambda)$ where $G$ is an oriented graph and $\Lambda$ is a mapping from the arcs of $G$ to elements of a group. A (not necessarily directed) cycle $C$ is called non-null if for any cyclic ordering of the arcs in $C$, the group element obtained by `adding' the labels on forward arcs and `subtracting' the labels on reverse arcs is not the identity element of the group.
Non-null cycles in group labeled graphs generalize several well-known graph structures, including odd cycles.

%
%
%
%
A family $\cF$ of graphs has the {\erdosposa} property, if for every integer $k$ there exists an integer
$f(k, \cF)$ such that every graph $G$ contains either $k$ vertex-disjoint subgraphs each isomorphic to a
graph in $\cF$ or a set at most $f(k, \cF)$ vertices intersecting every subgraph isomorphic to
a graph in $\cF$.

In this paper, we prove that non-null cycles on Group Labeled Graphs have the $\frac{1}{2}$-integral {\erdosposa} property. That is, there is a function $f:{\mathbb N}\to {\mathbb N}$ such that for any $k\in {\mathbb N}$, any group labeled graph $(G,\Lambda)$ has a set of $k$ non-null cycles such that each vertex of $G$ appears in at most two of these cycles or there is a set of at most $f(k)$ vertices that intersects every non-null cycle. Since it is known that non-null cycles do not have the {\erdosposa} property in general, a $\frac{1}{2}$-integral {\erdosposa} result is the best one could hope for.
%
\end{abstract}

\newpage
\pagestyle{plain}
\setcounter{page}{1}

\section{Introduction}

A family $\cF$ of graphs has the {\erdosposa} property, if for every integer $k$ there exists an integer
$f(k, \cF)$ such that every graph $G$ contains either $k$ vertex-disjoint subgraphs each isomorphic to a
graph in $\cF$ or a set at most $f(k, \cF)$ vertices intersecting every subgraph isomorphic to
a graph in $\cF$. This property is so named due to the famous result of Erd\H{o}s and P\'{o}sa \cite{Erdosposa65} who showed that the family of cycles in undirected graphs has this property. Since then, there has been a very long line of research identifying families of graphs which have the {\erdosposa} property.

Reed, Robertson, Seymour and Thomas \cite{ReedRST96} proved that the family of directed cycles also have this property, settling a 25 year old conjecture of Younger. Reed and Shepherd \cite{ReedS96} subsequently studied the same property for directed cycles in \emph{planar} digraphs and showed that the function $f(k,\cF)$ can be bounded as $\bigoh(k \log k \log \log k)$. There has also been a significant amount of research for families $\cF$ which are not only composed of cycles. Robertson and Seymour \cite{RobertsonS86}, in their Graph Minors project showed that when $\cF_H$  is a family of graphs containing a fixed graph $H$ as a minor, then $\cF_H$ has the {\erdosposa} property if and only if $H$ is planar. The bound on the function $f(k,\cF_H)$ given by Robertson and Seymour in this work was a large exponential and this was recently improved to $\bigoh(k \log^{\bigoh(1)}(k))$ (for a fixed $H$) by Chekuri and Chuzhoy \cite{ChekuriC14}.

One of the most well-studied class of graphs that \emph{does not} admit this property is the class of odd cycles. Lov\'{a}sz and Schrijver  showed that odd cycles in certain special grid-like graphs called \emph{Escher walls} do not have the {\erdosposa} property (see \cite{Reed99}). However, Reed \cite{Reed99} showed that the Escher walls are essentially the only obstructions and showed that (a) the {\erdosposa} property does hold for odd cycles on graphs that exclude a large enough Escher wall and (b) odd cycles have the $\frac{1}{2}$-integral {\erdosposa} property.  A family $\cF$ of graphs has the $\frac{1}{2}$-integral {\erdosposa} property, if for every integer $k$ there exists an integer
$f(k, \cF)$ such that every graph $G$ contains either $k$ subgraphs each isomorphic to a
graph in $\cF$ with the property that any vertex of $G$ appears in at most 2 of these subgraphs or a set at most $f(k, \cF)$ vertices intersecting every subgraph isomorphic to
a graph in $\cF$. There is an active line of research on such \emph{fractional} {\erdosposa} results for graph classes that do not admit the {\erdosposa} property \cite{KakimuraKK12,Huynharxiv}.

\interfootnotelinepenalty=10000

The family of non-null cycles in group-labeled graphs provide a natural way of generalizing the class of odd cycles in undirected graphs. A group-labeled graph is a pair $(G,\Lambda)$ where $G$ is an oriented graph and $\Lambda$ is simply a labeling of the arcs of $G$ by elements of some group $\Gamma$. A cycle in the undirected graph underlying $G$ is called \emph{non-null} if for any cyclic ordering of the arcs in this cycle, the group element obtained by `adding' the labels on forward arcs and `subtracting' the labels on backward arcs is not $1_\Gamma$, the identity element of the group $\Gamma$. We refer the reader to the section on preliminaries for the formal definitions. It is easy to see that if we take an arbitrary undirected graph $G$, an arbitrary orientation of the edges and construct the labeling $\Lambda$ by labeling each arc with the element $1$ from the group $\mathbb{Z}_2$, then the set of odd cycles in $G$ are precisely the set of non-null cycles in $(G,\Lambda)$. As a result, it follows that the family of non-null cycles in group-labeled graphs do not have the {\erdosposa} property in general. However, there are several results describing additional restrictions under which this family does have the required property. Kawarabayashi and Wollan \cite{KawarabayashiW06} showed that if the undirected graph underlying $G$ is sufficiently highly connected, then non-null cycles even for non-abelian group-labeled graphs have the {\erdosposa} property. Similarly, Wollan \cite{Wollan11} showed that if the graph is labeled with an abelian group without an element of order 2, then non-null cycles do have the {\erdosposa} property.  Huynh et al. \cite{Huynharxiv} in their original paper generalized this result for `doubly non-zero cycles' in graphs labeled with abelian groups. 
However, their requirement that the groups involved need to be mapped to abelian groups means that the result cannot be applied to the general case~\footnote{It has been brought to our attention that subsequent to our paper, they have managed to extend their result to the non-abelian case.  However, they do not provide a bound on their {\erdosposa} function, which appears to be worse than the bound proved in this paper.  }. 
In this work, we fill in this gap by showing that the $\frac{1}{2}$-integral {\erdosposa} result of Reed \cite{Reed99} for odd cycles can be generalized to the case of non-null cycles. 

\begin{restatable}{theorem}{maintheorem}\label{thm:erdos-posa}
	There is a function 
	$\tau: {\mathbb N} \to {\mathbb N}$ such that 
	$\tau(x)=2^{2^{\Theta(x\log x)}}$ and  for all $k\in {\mathbb N}$ and  groups $\Gamma$, for any $\Gamma$-labeled graph $(G,\Lambda)$, one of the following statements hold.
	\begin{itemize}
		\item $G$ contains a set of $k$ non-null cycles such that any vertex of $G$ is contained in at most 2 of these cycles or
		\item there is a set of $\tau(k)$ vertices that intersects every non-null cycle in $G$.
	\end{itemize}
	
	\noindent
	Furthermore, there is an algorithm that, given a $\Gamma$-labeled graph $(G,\Lambda)$ and an integer $k$, runs in time $2^{2^{\Theta(k\log k)}} n^{\bigoh(1)}$ and either returns a $\frac{1}{2}$-integral $k$-packing of non-null cycles or a set $X$ of size at most $\tau(k)$ which intersects every non-null cycle in $G$. 
	\end{restatable}


%
%
%


\begin{table}[t]
\centering
\setlength{\tabcolsep}{4pt}
{\footnotesize
\begin{tabular}{l  l}
\toprule
Family $\cF$          &  Reference       \\
\midrule
{\sc Undirected cycles} &   Erd\H{o}s and P\'{o}sa \cite{Erdosposa65} \\
{\sc Directed cycles} &   Reed, Robertson, Seymour and Thomas \cite{ReedRST96} \\

{\sc Odd-cycles} &   Reed \cite{Reed99}    \\
{\sc Odd-cycles (highly connected graphs)} &   Thomassen \cite{Thomassen01}, Rautenback and Reed \cite{RautenbachR01}, Joos \cite{Joos}     \\
{\sc Non-null cycles (highly connected graphs) } &   Kawarabayashi and Wollan \cite{KawarabayashiW06}\\
{\sc  Non-null Cycles ({\rm Abelian Groups}) } &      Wollan \cite{Wollan11}\\
{\sc $S$-cycles}  &  Kakimura et al. \cite{KakimuraKM11}  and Pontecorvi and Wollan \cite{PontecorviW12}             \\
{\sc Odd/Even $S$-cycles} &     Kakimura, Kawarabayashi and Kobayashi \cite{KakimuraKK12}     \\


{\sc  Long Cycles } &      Bruhn, Heinlein and Joos  \cite{BruhnJS}\\
{\sc Doubly non-zero cycles ({\rm Abelian Groups})} &   Huynh, Joos and Wollan \cite{Huynharxiv} \\

\bottomrule
\end{tabular}
		
	}
	\caption{\label{fig:vertexresults} Summary of some related {\erdosposa} results.}
\end{table}

\paragraph{Our techniques.} At a high level we follow the same  approach used in the proofs of almost all fractional  {\erdosposa} theorems. However, our implementation introduces interesting techniques for treewidth reduction.   We divide the cases into graphs with low treewidth and high treewidth. In the former case, the fact that the forbidden structure we are looking for is connected immediately implies the existence of an integral {\erdosposa} function. In the latter case, we show that the graph contains an \emph{irrelevant} vertex which can be computed in {\FPT} time. For this, we develop a new `treewidth-reduction' lemma which is similar in spirit to the Treewidth Reduction Theorem of Marx et al. \cite{MarxOSR13}. They showed that the set of all minimal $s$-$t$ separators of  size at most some $k\in {\mathbb N}$ is contained in a graph of bounded treewidth where the dependence of the treewidth on $k$ is at least doubly exponential. We on the other hand, obtain a result that essentially states that if the graph has treewidth exceeding a bound in $2^{\bigoh(k)}$, then there is a vertex which is \emph{not} part of any minimal $s$-$t$ separator of size at most $k$ and furthermore, this vertex can be computed in time $2^{\bigoh(k)}n^{\bigoh(1)}$. Since one can easily construct a graph where the treewidth is $2^{\Omega(k)}$ and all vertices are part of a minimal $s$-$t$ separator of size at most $k$, our result is asymptotically optimal in this respect. However, we point out that the improved bound in our treewidth reduction lemma comes at the cost of a bigger dependence of our algorithm on the input size as opposed to the approach of Marx et al. which allows one to design \emph{linear-time} {\FPT} algorithms. Furthermore, our proof technique can be easily adapted to show that non-null cycles on \emph{planar} graphs do admit the {\erdosposa} property and the function $f(k,\cF)=2^{\Theta(k \log k)}$.

 \paragraph{Related work on Feedback Set problems in Group Labeled graphs.} Guillemot \cite{Guillemot11a} initiated the study of the parameterized complexity of the problem of deciding whether a given group labeled graph has a set of at most $k$ vertices or arcs that intersect every non-null cycle. He showed that the vertex-version of the problem is fixed-parameter tractable ({\FPT}) parameterized by $k$ and the size of the group while the edge-version is {\FPT} parameterized by $k$ alone.
  His results were improved upon by Cygan et al. \cite{CyganPP16} who gave the first {\FPT} algorithm for the vertex-version of this problem parameterized by $k$ alone and subsequently by Wahlstr\"{o}m \cite{Wahlstrom14} who obtained an {\FPT} algorithm with optimal asymptotic dependence on $k$.






%
\section{Preliminaries}

\paragraph{{\sf Group Labeled Graphs}.}
Let $\Gamma$ be a group. A $\Gamma$-labeled graph is a pair $(G,\Lambda)$ where $G$ is a digraph and $\Lambda:A(G)\to \Gamma$. For every $a=(u,v)\in A(G)$, we define $\Lambda((v,u))=\Lambda((u,v))^{-1}$. For a digraph $G$, we denote by $\tilde G$ the underlying undirected digraph.
 Let $C$ be a cycle in the undirected graph underlying $G$. 
Let $P=v_1,\dots, v_\ell$ be a path in $\tilde G$. We denote by $\Lambda(P)$ the element $\Lambda((v_1,v_2))\cdot \Lambda((v_2,v_3))\cdot \Lambda ((v_{\ell-1},v_\ell))$. Let $C=v_1,\dots, v_\ell,v_1$ be a cycle in $\tilde G$. We denote by $\Lambda(C)$ the element $\Lambda((v_1,v_2))\cdot \Lambda((v_2,v_3))\cdot \Lambda ((v_{\ell-1},v_\ell)) \cdot \Lambda ((v_{\ell},v_1)) $. We call $C$ non-null if $\Lambda(C)\neq 1_\Gamma$. Note that even though different choices of the start vertex $v_1$ in the same cycle may lead to different values for $\Lambda(C)$, it is easy to see that if for one choice of $v_1$ the value of $\Lambda(C)$ is not $1_\Gamma$ then for no choice of $v_1$ is it $1_\Gamma$. Whenever we refer to a path or cycle in the graph $G$, we will be talking about the path or cycle in $\tilde G$. For a pair of vertices $u,v$ we call a $u$-$v$ path $P$ in $\tilde(G)$ a \emph{non-null} $u$-$v$ path if $\Lambda(P)\neq 1_\Gamma$. We call $P$ an \emph{identity} path if $\Lambda(P)=1$. For a set $A\subseteq V(G)$ we say that a path in $\tilde G$ is a non-null $A$-path if it is a non-null $u$-$v$ path for some $u,v\in A$.
A set $X\subseteq V(G)$ that intersects every non-null cycle in $G$ is called a group feedback vertex set (\emph{gfvs}) of $G$. A set ${\cal C}=\{C_1,\dots, C_k\}$ of $k$ vertex-disjoint non-null cycles is called a $k$-packing of non-null cycles or an \emph{integral} $k$-packing of non-null cycles. If every vertex of $G$ appears in at most 2 cycles in the set $\cal C$, then we call $\cal C$ a $\frac{1}{2}$-integral $k$-packing of non-null cycles. 


\begin{definition}
	Let $(G,\Lambda)$ be a $\Gamma$-labeled graph with labeling function $\Lambda$. Let $\lambda:V(G)\to \Gamma$. We say that $\lambda$ is a \textbf{consistent labeling} if for each arc $(u, v)\in A(G)$ we have $\lambda(v)=\lambda(u)\cdot \Lambda(a)$.
\end{definition}

\begin{lemma}{\sc \cite{Guillemot11a,CyganPP16}} Let $(G,\Lambda)$ be a $\Gamma$-labeled graph and suppose that $G$ has no non-null cycles. Then, $G$ has a consistent labelling.
\end{lemma}

%

\begin{observation}
\label{obs:trivial}
Let $(G,\Lambda)$ be a $\Gamma$-labeled graph. Then, $(G,\Lambda)$ has no non-null cycles if and only if $G$ has a closed walk $W$ such that $\Lambda(W)\neq 1_\Gamma$. 
\end{observation}

\begin{proposition}{\rm \cite{ChudnovskyGGGLS06, ChudnovskyCG08}}\label{prop:gallai} Let $(G,\Lambda)$ be a $\Gamma$-labeled graph and let $S\subseteq V(G)$. For any $k\in {\mathbb N}$ one of the following holds:
 
 \begin{itemize}
 	\item there exist vertex-disjoint non-null $S$-paths $P_1,\dots, P_k$.
 	\item there exists $X\subseteq V(G)$ with $|X|\leq 2k-2$ such that $X$ hits all non-null $S$-paths.
 \end{itemize}

 \noindent
 	Furthermore there is an algorithm that, given $(G,\Lambda)$, $k$ and $S$ runs  in time $\bigoh(|V(G)|^6)$ and computes either the paths $P_1,\dots, P_k$ or the set $X$.
 \end{proposition}
 
 \begin{definition}
	For a $\Gamma$-labeled graph $(G,\Lambda)$, the result of {\bf untangling} $\Lambda$ around a vertex $v$ with group element $g\in \Gamma$ is the graph $(G,\Lambda')$ where $\Lambda'$ is defined as follows. For every arc $a$ which is not incident on $v$, $\Lambda'(a)=\Lambda(a)$. For every arc $(v,u)\in A(G)$, we set $\Lambda'(v,u)=g\cdot \Lambda(v,u)$ and for every arc $(u,v)\in A(G)$, we set $\Lambda'(u,v)=\Lambda(u,v)\cdot g^{-1}$.

	 Let $A\subseteq V(G)$ such that $G[A]$ has no non-null cycles. Let $\lambda:V(G[A])\to \Gamma$ be a consistent labeling of $G[A]$. We denote by $\untangle(G,\Lambda,A)$ the graph obtained from $G$ by untangling $\Lambda$ around every vertex $v\in A$ with the group element $\lambda(v)$ (in an arbitrary order). We will not explicitly refer to $\lambda$ when performing untangling as it can be chosen to be any consistent labeling of $G[A]$.\end{definition}
	
	As a direct consequence of the definition of the untangling operation, we have the following.

\begin{proposition}
	\label{prop:untangling} Let $(G,\Lambda)$ be a $\Gamma$-labeled graph and let $X\subseteq V(G)$ such that $(G[X],\Lambda)$ has no non-null cycles. Then, there is a labelling $\Lambda'$ for $A(G)$ such that the set of non-null cycles in $(G,\Lambda)$ and those in $(G,\Lambda')$ are the same and for every arc $a=(u,v)\in A(G)$ such that $u,v\in X$, we have that $\Lambda'(a)=1_\Gamma$.
\end{proposition}

The statement follows by picking an arbitrary consistent labeling $\lambda$ for $G[X]$ and untangling $\Lambda$ around each $x\in X$ with the group element $\lambda(x)$.

\paragraph{{\sf Tree-decompositions}.}

A \emph{tree decomposition} of a graph~$G$ is a pair~$(T, \chi)$, where~$T$ is a tree and~$\chi : V(T) \rightarrow 2^{V(G)}$ is a function such that the following conditions are satisfied:
\begin{enumerate}\setlength\itemsep{-.7mm}
	\item For each edge~$uv \in E(G)$ there is a node~$b \in V(T)$ such that~$\{u,v\} \subseteq \chi(b)$.
	\item For each~$v \in V(G)$ the nodes~$\{b \mid v \in \chi(b)\}$ induce a non-empty connected subtree of~$T$.
\end{enumerate}

The sets $\chi(b)$ for $b \in V(T)$ are called {\em bags} of the tree decomposition. We will assume that the reader is familiar with tree decompositions and their basic properties. 
For an introduction to tree decompositions, see~\cite{diestelbook}. 
The tree-width of a graph $G$ is the minimum width of any tree-decomposition of $G$.

\paragraph{{\sf Minors}.} We follow the notation and terminology from \cite{GeelenGRSV09}. Let $H$ and $G$ be graphs.
An $H$-expansion in $G$ is a function $\eta$ with domain $V(H)\cup E(H)$ such that for every $v\in V(H)$, $\eta(v)$ is a subgraph of $G$ that is a tree, for every $u,v\in V(H)$, $\eta(u)$ and $\eta(v)$ are vertex disjoint, for every $e=(u,v)\in E(H)$, $\eta(e)$ is an edge $f\in E(G)$ such that $f$ is incident in $G$ with both $\eta(u)$ and $\eta(v)$. The trees $\{\eta(v):v\in V(H)\}$ are called the \emph{supernodes} of the expansion. For each supernode $\eta(v)$, we fix an arbitrary vertex, called the \emph{center} and denote it by $\gamma(v)$.
We use $\bigcup \eta$ to denote the subgraph of $G$ given by the union of the nodes and use $\eta(u,v)$ to denote the edge between $\eta(u)$ and $\eta(v)$ corresponding to the edge $(u,v)\in E(H)$.  The graph $K_\ell$ denotes the clique on $\ell$ vertices.
%
%
%
  We will also require the following simple observation regarding the interaction of a clique-expansion with  blocks of the graph. Recall that a block of a graph is a maximal 2-vertex-connected subgraph of the graph.

\begin{observation}\label{obs:unique_block}
Let $G$ be a graph, $\ell\in {\mathbb N}$ and  $\eta$ be a $K_\ell$-expansion in $G$. For any set $X\subseteq V(G)$ of at most $\ell-1$ vertices, there is a unique component of $G-X$ that completely contains every $\eta(v_i)$ disjoint from $X$ and further, there is a unique block of $G-X$ that intersects every $\eta(v_i)$ disjoint from $X$.
\end{observation}

\begin{proof}
	The first part of the statement follows immediately from the fact that for any pair $i\neq j\in [\ell]$ with $\eta(v_i)$, $\eta(v_j)$  disjoint from $X$, the edge $\eta(v_iv_j)$ ensures that both these sets are contained in the same component of $G-X$.
	We now argue the second part which is stronger. That is,  there is a unique block of $G-X$ that \emph{intersects} every $\eta(v_i)$ disjoint from $X$. Observe that if there are exactly 2 supernodes $\eta(v_i)$ and $\eta(v_j)$ disjoint from $X$, then the unique block of $G-X$ containing the edge $\eta(v_iv_j)$ satisfies the statement. Hence, we may assume that there are at least 3 supernodes of $\eta$ disjoint from $X$.

	 Now, suppose that for a pair $i\neq j\in [\ell]$ with $\eta(v_i)$, $\eta(v_j)$ disjoint from the $X$, there is a block of $G-X$, say $\cB$ such that $\cB$ intersects $\eta(v_i)$ and is disjoint from $\eta(v_j)$. Furthermore, we may assume without loss of generality that $\cB$ also intersects the supernode $\eta(v_r)$ for some $r\neq i,j$. This is due to the edge $\eta(v_rv_i)$. But now, we can use the edges $\eta(v_iv_j)$ and $\eta(v_rv_j)$ to construct a pair of $\eta(v_j)$-$\cB$ paths which are internally vertex-disjoint and have distinct endpoints in $\cB$. But this implies that some vertex of $\eta(v_j)$ is in the  block $\cB$, a contradiction to our assumption that $\eta(v_j)$ is disjoint from $\cB$. Hence we conclude that there is a unique block of $G-X$ that intersects every $\eta(v_i)$ disjoint from $X$.	 
	 	\end{proof}

\begin{observation}\label{obs:non-null-cycle-block}
Let $(G,\Lambda)$ be a $\Gamma$-labeled graph and let $\cB$ be a block in this graph. If $\cB$ contains a non-null cycle, then for every $u,v\in \cB$, there is a non-null cycle that contains both $u$ and $v$.
\end{observation}

\paragraph{{\sf Linkages}.} A set of vertex disjoint paths is called  a {\em linkage}. Let $X$ and $Y$ be vertex sets and  $\cP$ be a set of vertex-disjoint $X$-$Y$ paths which are internally disjoint from $X\cup Y$. If every vertex in $X\cup Y$ appears as the endpoint of a path in $\cP$, we call $\cP$ an $X$-$Y$ linkage. If $X$ and $Y$ are sets in a group labeled graph and each path in $\cP$ is a non-null path, then we call $\cP$ a non-null $X$-$Y$ linkage. Note that a single vertex in $X\cap Y$ is considered to be an $X$-$Y$ path. If $|Y|\geq |X|$ and there is a $Y'\subseteq Y$ such that there is an $X$-$Y'$ linkage then we say that $X$ is linked to $Y$.

\begin{definition}

A set $X\subseteq V(G)$ is $p$-linked in $G$ if $|X|\geq p$ and for all subsets $X_1,X_2\subseteq X$ with $|X_1|=|X_2|\leq p$, there is an $X_1$-$X_2$ linkage in $G$. The set $X$ is said to be well-linked if it is $\frac{|X|}{2}$-linked.

\end{definition}


 \begin{observation}\label{obs:well-linked-separator}
Let $G$ be an undirected graph, $r\in {\mathbb N}$ and let $Z$ be a well-linked set in $G$. For any set $X\subseteq V(G)$ such that $|X|< |Z|/2$ and any connected component $C$ of $G-X$, either $|C\cap Z|\leq |X|$ or $(Z\setminus (C\cup X))\leq |X|$. 
\end{observation}


\paragraph{{\sf Separations}.}  Given a graph $G = (V, E)$, an ordered pair $(A, B)$ of vertex sets such that
$A \cup B = V(G)$ is called a \emph{separation} of $G$ if there is no edge $(u,v)$ with $u\in A\setminus B$ and  $v\in B\setminus A$. The sets $A$ and $B$ are called the \emph{sides} of this separation. A separation
$(A, B)$ such that neither $A \subseteq B$ nor $B \subseteq A$ is called a \emph{proper} separation. The order
of a separation $(A, B)$ is the cardinality of its separator $A \cap B$. A separation of order $k$ is called a $k$-separation.


\begin{definition}
	Let $k,p\in \mathbb N$ and let $(G,\Lambda)$  be  a $\Gamma$-labeled graph. We denote by $\ep(G,k,p)$ the predicate : $G$ has either a $\frac{1}{2}$-integral $k$-packing or a gfvs of size at most $p$.
\end{definition}

\begin{definition}
	Let $k\in \mathbb N$ and let $(G,\Lambda)$ and $(G',\Lambda')$ be  $\Gamma$-labeled graphs where $G'$ is a subgraph of $G$. We say that $(G',\Lambda')$ is $k$-{\bf equivalent} to $(G,\Lambda)$  if for every $p\in {\mathbb N}$, if $\ep(G',k,p)$ holds then $\ep(G,k,p)$ holds. Again, we ignore the explicit reference to $\Lambda$ and $\Lambda'$ when clear from the context. 
\end{definition}


\begin{definition}
Let $(G,\Lambda)$ be a $\Gamma$-labeled graph and let $S\subseteq V(G)$. If $G[S]$ has no non-null cycles then we call $G[S]$ a {\bf clean subgraph}.
\end{definition}

\paragraph{{\sf Separators}.} Let $X$ and $Y$ be disjoint vertex sets in an undirected graph $G$. We say that a set $S\subseteq V(G)\setminus (X\cup Y)$ is an $X$-$Y$ {\bf separator} if it intersects all $X$-$Y$ paths. We denote by $R(X,S)$ the set of vertices in the components of $G-S$ that intersect $X$.

\begin{definition}An $X$-$Y$ separator $S$ is called {\bf important} if there is no other $X$-$Y$ separator $S'$ such that $|S'|\leq |S|$ and $R(X,S')\supset R(X,S)$. If such a separator $S'$ does exist, then it is said to {\bf dominate} the separator $S$.
\end{definition}

\begin{proposition}\label{prop:impsep}{\sc \cite{CyganFKLMPPS15}} The number of important $X$-$Y$ separators of size at most $k$ is bounded by $4^k$ and these can be enumerated in time $\bigoh(4^k n^{\bigoh(1)})$.	
\end{proposition}


\paragraph{{\sf Remark}.} If a subgraph is known to be clean, then we occasionally treat it as a simple undirected graph by forgetting the edge labels and directions. When it is clear from the context that the subgraph under consideration is clean, we will not explicitly mention this operation of removing the group labels and directly invoke lemmas that normally only apply to undirected graphs. 
Finally, we remark that we have not attempted to optimize the constants in our lemmas and in fact we have chosen loose bounds (preserving the asymptotics) whenever possible to keep the proofs simpler.

\section{Treewidth reduction lemma}
In this section, we begin by presenting our treewidth reduction lemma.  We prove a more general statement than the one mentioned in the introduction, in terms of multiple terminals and node multiway cuts instead of two terminals $s$ and $t$ and $s$-$t$ seperators.

\begin{definition}
	Let $G$ be an undirected graph and $T\subseteq V(G)$. Let $\cP=\{P_1,\dots, P_\ell\}$ be a partition of $T$ where $\ell>1$. We say that $X\subseteq V(G)\setminus T$ is a $(\cP,T)$-{\bf multiway cut} if $G-X$ has no $P_i$-$P_j$ path for any $1\leq i\neq j\leq \ell$. We say that $X$ is a $T$-{\bf multiway cut} if there is a partition of $T$, say $\cP=\{P_1,\dots, P_\ell\}$ such that it is a $(\cP,T)$-multiway cut.
	We say that $X$ is a minimal $(\cP,T)$-multiway cut ($T$-multiway cut) if there is no strict subset which is also a $(\cP,T)$-multiway cut ($T$-multiway cut).
\end{definition}

\begin{lemma}\label{lem:tw_reduction_one}
For any  undirected graph $G$, an integer $t>1$ and a vertex set $T$ of size at most $t$, 
if $G$ contains a well-linked set $Z$ of size at least $7t$ which is disjoint from $T$, 
 then there is a set $\cX\subseteq Z$ of at most 
$t^{6t}$ vertices such that no vertex $v\in Z\setminus \cX$ is part of a minimal $T$-multiway cut of size at most $t$.  Furthermore, there is an algorithm that, given $G,T,Z$ and $t$, runs in time $2^{\bigoh(t\log t)}n^{\bigoh(1)}$ and computes the set $\cX$.	
	
\end{lemma}

\begin{proof}

We first construct a graph $G^\star$ as follows. We start with the graph $G$ and add a new  vertex $q^\star$ and  a path of length 2 between $q^\star$ and every vertex of $Z$. This completes the description of the graph $G^\star$. For every vertex $x\in Z$,  we refer to the unique vertex adjacent to $x$ and $q^\star$ as $x^\star$. We refer to the set of new degree-2 vertices adjacent to $Z$ as $Z^\star$.

For each partition of $T$,  $\cP=\{P_1,\dots, P_\ell\}$, where $\ell>1$, we compute $\ell+1$ vertex sets $\cX^{\cP},\cX^{\cP}_1,\dots, \cX^\cP_\ell$.
For each $i\in [\ell]$, the set $\cX^{\cP}_i$ is defined in the following way. Initially, $\cX^\cP_i=\emptyset$. For every important $P_i$-$(q^\star\cup (T\setminus P_i))$ separator $Q$ of size at 
most $2t$ in the graph $G^{\cP}_i$ (if it exists), we add the vertices in $Z\cap R_{G^\star}(P_1,Q)$ to the set $\cX^{\cP}_i$.
Finally, we define the set $\cX^{\cP}$ to be $\bigcup_{i\in [\ell]} \cX^{\cP}_{i}$ and the set $\cX$ claimed in the statement of the lemma as $\cX=\bigcup_{\cP:\texttt{ partition of }T} \cX^{\cP}$. We now argue the claimed bound on the size of $\cX$. For that, we prove a bound on the size of the set $\cX^\cP$ for a specific partition $\cP$.

\begin{claim}
	$|\cX^\cP|\leq 16^t\cdot 2t$.
\end{claim}

\begin{proof}
Due to Observation \ref{obs:well-linked-separator}, we know that for any $P_i$-$(q^\star\cup (T\setminus P_i))$ separator $Q$ of size at 
most $2t$ in the graph $G^{\cP}_i$, one of the sets $Z\cap R_{G^\star}(P_1,Q)$ or $Z\cap NR_{G^\star}(P_1,Q)$ has size at most $2t$.  Suppose that that $Z\cap NR_{G^\star}(P_1,Q)$ has size at most $2t$. Since $|Z|\geq 7t$, it must be the case that $Z\cap R_{G^\star}(P_1,Q)$ has size at least $3t$ (after removing the intersection with $Q$ also). But observe that for every vertex $v\in Z\cap R_{G^\star}(P_1,Q)$, it must be the case that $v^\star \in Q$. But since the size of $Q$ is at most $2t$, this leads to a contradiction. Hence, we conclude that that for any $P_i$-$(q^\star\cup (T\setminus P_i))$ separator $Q$ of size at 
most $2t$ in the graph $G^\star$, the size of the set $Z\cap R(P_1,Q)$ is at most $2t$. 
Since the number of important $P_i$-$(q^\star\cup (T\setminus P_i))$ separators of size at 
most $2t$ is bounded by $4^{2t}$ (by Proposition \ref{prop:impsep}), the claim follows.
\end{proof}

Due to the above claim, the size of $\cX$ is bounded by $t^t\cdot 4^{2t} \cdot 2t\leq  t^{6t}$ since $t>1$.
Having proved the size bound on $\cX$, it remains to argue that no vertex $v\in Z\setminus \cX$ is part of a minimal $T$-multiway cut of size at most $t$ in $G$. Suppose to the contrary that for some partition of $T$, say $\cP=\{P_1,\dots, P_\ell\}$, there is a minimal $(\cP,T)$-multiway cut $S$ of size at most $t$ which also contains the vertex $v$. Since $S$ is minimal, we may assume without loss of generality that $G-S$ has two connected components $C_1,C_2$ such that $C_1$ intersects $P_1$, $C_2$ intersects $P_2$ both $C_1$ and $C_2$ are adjacent to the vertex $v$. By Observation \ref{obs:well-linked-separator}, we know that one of $C_1$ or $C_2$ contains at most $t$ vertices of $Z$. Without loss of generality, suppose that $C_1$ is such a component and let $Q=C_1\cap Z$ and $S=N(C_1)$. 
Let $Q^\star$ denote the set of degree-2 vertices adjacent to the vertices in $Q$. Let $S'=S\cup Q^\star$. Observe that $S'$ is now a \emph{minimal} $P_1$-$(q^\star\cup (T\setminus P_i))$ separator of size at most $2t$. Since $S'$ contains the vertex $v$, it cannot be an important $P_1$-$(q^\star\cup (T\setminus P_i))$ separator. Hence, there is a $P_1$-$(q^\star\cup (T\setminus P_i))$ separator $J$ that dominates $S'$ and furthermore, $J$ does not contain $v$. However, this implies that $v\in R_{G^\star}(P_1,J)$ and hence, we would have added $v$ to the set $\cX^\cP_1$, a contradiction. Therefore, we conclude that no vertex $v\in Z\setminus \cX$ is part of a minimal $T$-multiway cut of size at most $t$ in $G$. Finally, the running time bound follows from that of Proposition \ref{prop:impsep} and the fact that the number of partitions of a set of size at most $t$ is upper bounded by $t^t$. This completes the proof of the lemma.
	\end{proof}

Observe that when the above lemma is specialized to the case of minimal $s$-$t$ separators of size at most $t$, the step where we go over all partitions of $T$ is redundant and hence we obtain a bound of $2^{\bigoh(t)}$ on the size of the set $\cX$. We now show how to use the lemma above to detect an irrelevant vertex when certain conditions are satisfied.

\begin{lemma}\label{lem:linkage2}
	For every $p\in \mathbb N$ and $\Gamma$-labeled graph $G$, if $G$ has  a separation $(A,B)$  such that \begin{enumerate}
	\item $G[A]$ is clean,  
	\item $1<|A\cap B|\leq p$ and
	\item $G[A]$ contains a well-linked set $Z$ of size $> 2^{p}\cdot p^{6p}$ disjoint from $A\cap B$,
 \end{enumerate}

\noindent
then there is a vertex $v$ in $A\setminus B$ such that 
the graph $G-v$ is $k$-equivalent to the graph $G$ for any $k\in {\mathbb N}$ and given $G,p,Z$ and the separation $(A,B)$, the vertex $v$ can be computed in time $2^{\bigoh(p\log p)}n^{\bigoh(1)}$.

\end{lemma}

\begin{proof} 	 
Let $X=A\cap B$ and $G'=G[A]$. We assume without loss of generality that the labels of the arcs of $G$ have been untangled around the vertices in $A$ and hence any arc with both endpoints in $A$ is labeled with the identity element $1_\Gamma$.
For each $J\subset X$, we apply Lemma \ref{lem:tw_reduction_one} on the graph $G'-J$ with the terminal set $T=X\setminus J$, the well-linked set $Z$ and $t=|X\setminus J|$ and compute a set $\cZ_J\subseteq Z$ such that any vertex of $Z\setminus \cZ_J$ is not part of a $(X\setminus J)$-multiway cut of size at most $|X\setminus J|$ in the graph $G'-J$. We define the set $\cZ=\bigcup_{J\subset X}\cZ_J$. It follows from Lemma \ref{lem:tw_reduction_one} that $|\cZ|\leq 2^{p}\cdot p^{6p}$. But we  know that $|Z|> 2^{p}\cdot p^{6p}$ by the premise of the lemma. Hence, the set $Z\setminus \cZ$ is non-empty. Let $v$ be an arbitrary vertex in this set. We now argue that $G-v$ is $k$-equivalent to $G$ and hence we can simply return $v$ as the required vertex.


   Since $G-v$ is a subgraph of $G$, for any $k\in {\mathbb N}$, if $G-v$ contains a $\frac{1}{2}$-integral $k$-packing then so does $G$. Now, let $Q$ be an arbitrary minimum gfvs of $G-v$. If $G$ also has a gfvs of size $|Q|$, then we are done. Suppose that this is not the case. In particular, $Q$ is not a gfvs of $G$.  This implies that there is a non-null cycle $C$ in $G-Q$ which intersects $v$. Furthermore, since $G[A]$ is clean and $v\in A\setminus B$, it must be the case that the cycle $C$ intersects $X$ in at least 2 vertices. Let $\alpha$ and $\beta$ be two vertices in $C\cap X$ such that the subpath of $C$ between them is internally disjoint from $X$ and  contains the vertex $v$. Let this subpath be $P$, let $J=Q\cap X$ and $G'=G[A]-J$. We now claim the following.

   \begin{claim}
   	There is no $\alpha$-$\beta$ path in $G'-v-(Q\cap A)$.
   \end{claim}
   
   \begin{proof}   Suppose that this is not the case and there is an $\alpha$-$\beta$ path $P'$ in $G'-v-(Q\cap A)$. But this implies that $P$ is an identity path in $G$ that is disjoint from $Q\cup \{v\}$. Hence, it is also present in $(G-v)-Q$. Also recall that the $\alpha$-$\beta$ path $P$ which is a subpath of $C$ is also an identity path since it is contained in the clean graph $G[A]$. We can now construct a closed walk $C'$ starting from $C$ by simply replacing the subpath $P$ with the subpath $P'$. Since both paths are identity paths, it follows that $C'$ is a non-null walk in the graph $(G-v)-Q$, which is a contradiction to our assumption that $Q$ is a gfvs for the graph $G-v$.
   This completes the proof of the claim.
   \end{proof}
   
     Let $\cP={P_1,\dots, P_r}$ denote the partition of $X\setminus J$ induced by the connected components of $G'-(Q\cap A)=(G[A]-J)-(Q\cap (A\setminus B))$.  Without loss of generality, we assume that $\alpha,\beta\in P_1$. Now, let $\cP'={P_1',\dots, P_r'}$ denote the partition of $X\setminus J$ induced by the connected components of $G'-v-(Q\cap A)=(G[A]-J)-v-(Q\cap (A\setminus B))$. Due to the claim above, we know that $\alpha$ and $\beta$ lie in distinct sets of this partition. But this implies that $v$ is part of a minimal $(\cP',X\setminus J)$-multiway cut which is a (not necessarily strict) subset of $(Q\cap A)\cup \{v\}$.

    We now observe that $|Q\cap A|< |X|$.  If this were not the case, then $(Q\setminus A)\cup X$ is also a gfvs of $G$, contradicting our assumption that $G$ has no gfvs of size at most $|Q|$.
     Hence, we conclude that $v$ is part of a minimal $(\cP',X\setminus J)$-multiway cut of size at most $|X\setminus J|$ in the graph $G'-J$, a contradiction to the fact that $v$ was not in the set $\cZ_J$ computed using the algorithm of Lemma \ref{lem:tw_reduction_one}. Hence, we conclude that $G-v$ is $k$-equivalent to $G$. Finally, the time required to compute $v$ is dominated by the time required for $2^{p}$ executions of the algorithm of Lemma \ref{lem:tw_reduction_one},  which implies the claimed bound on the running time. This completes the proof of the lemma.	
\end{proof}

\section{The main theorem}

The proof of Theorem \ref{thm:erdos-posa} is divided into 3 cases. In Subsection \ref{subsec:lowtw}, we handle the low treewidth case while Subsection \ref{subsec:clique} handles  graphs having high treewidth and containing a large order clique expansion. In Subsection \ref{subsec:flatwall}, we handle the last case where the graph under consideration has large treewidth but no large order clique expansion. Finally, we combine the main lemmas of each subsection to give the full proof of Theorem \ref{thm:erdos-posa}. We assume throughout the rest of the paper that in any group-labeled graph we are dealing with, every arc is part of a non-null cycle since otherwise these can be detected and removed in polynomial time \cite{Guillemot11a}. For ease of presentation, we will when it is clear from the context that the graph under consideration is the underlying undirected graph of a group labeled graph $(G,\Lambda)$, we will refer to this graph as $G$ instead of $\tilde G$.

\subsection{Low treewidth graphs}
\label{subsec:lowtw}
\begin{lemma}\label{lem:bounded_treewidth_case}
	Let $(G,\Lambda)$ be a $\Gamma$-labeled graph, where $tw(G)\leq w$. For all $k\in {\mathbb N}$, either $G$ has $k$ vertex disjoint non-null cycles or has a set $X\subseteq V(G)$ of size at most $(k-1)(w+1)$ such that $G-X$ has no non-null cycles. Furthermore, there is an algorithm that, given $(G,\Lambda)$, $k$ and a tree-decomposition of $G$ of width at most $w$, runs in polynomial time and either returns a $k$-packing of non-null cycles or a gfvs of size at most $(k-1)(w+1)$.
\end{lemma}

\begin{proof} The proof is by induction on $k$. In the base case, $k=1$ and the statement holds vacuously. We now consider the case when $k>1$. 
 
 Consider an arbitrary tree-decomposition of $G$, $(T,\chi)$ of width at  most $w$. For every $b\in V(T)$, let $\alpha(b)$ denote the vertices in the bags corresponding to vertices of $T$ which are descendants of $b$.
 Let $T$ be rooted arbitrarily and let $b\in V(T)$ be a node such that $G[\alpha(b)]$ has a non-null cycle and $G[\alpha(b')]$ has no non-null cycles for any $b'$ which is a child of $b$. Consider the graphs $G_1=G[\alpha(b)]$ and $G_2=G-\alpha(b)$. By the induction hypothesis, we know that the graph $G_2$ has $k-1$ vertex disjoint non-null cycles or a gfvs of size at most $(k-2)(w+1)$. In the former case, these $k-1$ vertex disjoint non-null cycles along with an arbitrary non-null cycle in $G_1$ implies the presence of a set of $k$ vertex disjoint non-null cycles in $G$. In the latter case, an arbitrary gfvs of $G_2$ along with the set $\chi(b)$  is a gfvs of $G$. Hence, this implies the existence of a gfvs of $G$ whose size is at most $(w+1) +(k-2) (w+1) =(k-1)(w+1)$. This completes the induction step. Clearly, this argument can be made algorithmic, completing the proof of the lemma.
%
%
%
%
%
%
\end{proof}

\subsection{Graphs containing a large order clique-expansion}
\label{subsec:clique}

In this subsection, we describe our algorithm for the case when the graph under consideration has a clique-minor or clique-expansion whose size exceeds a specified bound depending on $k$. At a high level, we will show that either we can find the desired $\frac{1}{2}$-integral $k$-packing of non-null cycles or we can find a small set of vertices, say $X$ and a connected component of $G-X$, say $C$ which is clean (and so can be treated as an undirected graph) and  contains a large well-linked set. Finally, we show that we can either greedily add $X$ to the gfvs we are constructing or we can invoke Lemma \ref{lem:tw_reduction_one} on the graph $G[C\cup X]$ to find a vertex whose deletion from $G$ results in an equivalent instance.

\begin{lemma}\label{lem:fine_linkage1}
	There is a function $\rho(x)=2^{\Theta(x\log x)}$ such that 
	for every $k\in \mathbb N$ and $\Gamma$-labeled graph $G$, if $G$ has a $K_{\rho(k)}$-expansion $\eta$  and a separation $(A,B)$  such that \begin{enumerate}
	\item $G[A]$ is clean,  
	\item $1<|A\cap B|\leq 3k$ and
	\item $\bigcup \eta$ is contained in $A\setminus B$
 \end{enumerate}

\noindent
then there is a vertex $v$ in $A\setminus B$ such that 
the graph $G-v$ is $k$-equivalent to the graph $G$ and given $G,k, \eta$ and the separation $(A,B)$, the vertex $v$ can be computed in time $2^{\bigoh(k\log k)}n^{\bigoh(1)}$.

\end{lemma}

\begin{proof} 	We set $\rho(k)= 2^{3k}\cdot (3k)^{18k}+1$. Let $X=A\cap B$ and $G'=G[A]$. Since every supernode of $\bigcup \eta$ is contained in $A\setminus B$, it follows that $G[A]-X$ contains  a well-linked set $Z$ of size $\rho(k)$. That is, $|Z|> 2^{3k}\cdot (3k)^{18k}$.  We now invoke Lemma \ref{lem:linkage2} with input $G,k,Z$ and this separation $(A,B)$ to compute a vertex $v$ such that $G-v$ is $k$-equivalent to $G$.

Finally, the time required to compute $v$ is dominated by the time to execute the algorithm of Lemma \ref{lem:linkage2}, which implies the claimed bound on the running time. This completes the proof of the lemma.
	\end{proof}

We now prove the following lemma to demonstrate how one can get to a point such that the premises of the previous lemma are (almost) satisfied.

\begin{lemma}\label{lem:fine_linkage2}
	For every $k\in \mathbb N$, $\ell >6k^2$ and $\Gamma$-labeled graph $G$, if $G$ has a $K_{\ell}$-expansion $\eta^\star$ then it has a $\frac{1}{2}$-integral $k$-packing of non-null cycles or 
	 a separation $(A,B)$  such that \begin{enumerate}
	\item $G[A\setminus B]$ is clean,  
	\item $1<|A\cap B|\leq 3k$ and
	\item every supernode of $\eta^\star$ disjoint from $A\cap B$ is contained in $A\setminus B$.
 \end{enumerate}
 
 \noindent
 Furthermore, there is an algorithm that, given $G,k$ and $\eta^\star$, runs in polynomial time and returns either  a $\frac{1}{2}$-integral $k$-packing of non-null cycles or the separation $(A,B)$.

\end{lemma}

%
%
%
%
%
%

\begin{proof}
 We  divide $\eta^\star$ into $k$ vertex-disjoint $K_{\ell'}$-expansions where $\ell'=\ell/k$. We may assume without loss of generality that $\ell$ is divisible by $k$. If the graph induced on the vertex set of each individual $K_{\ell'}$-expansion contains a non-null cycle then we are done since there is an \emph{integral} $k$-packing of non-null cycles. Hence, we may assume the presence of a $K_{\ell'}$-expansion $\eta$ such that the graph $\bigcup \eta$ is clean. Going forward, we assume without loss of generality that $G$ has been untangled around every vertex in $\bigcup \eta$ and hence every arc in $\bigcup \eta$ is labeled with the identity element. Let $\gamma(v_1),\dots, \gamma(v_{\ell'})$ denote arbitrarily chosen centers for the supernodes of $\eta$. Let $S=\{\gamma(v_1),\dots, \gamma(v_{\ell'})\}$. We invoke Proposition \ref{prop:gallai} to either conclude that there is a set of $k$ vertex-disjoint non-null $S$-paths in $G$ or there is a set  of at most $2k$ vertices which intersects every non-null $S$-path. We now consider each case separately.
	
	In the first case, let $P_1,\dots, P_k$ denote the $k$ vertex-disjoint non-null $S$-paths and without loss of generality, let $x_i=\gamma(v_{2i-1})$ and $y_i=\gamma(v_{2i})$ be the endpoints of the path $P_i$ for each $i\in [k]$. Then, we define a set of closed walks $W_1,\dots, W_k$ with $W_i$ defined as the closed walk obtained by taking the sum of $P_i$ and the unique $x_i$-$y_i$ path contained within the tree induced by $\eta(v_{2i-1})\cup \eta(v_{2i})$. Since each $x_i$-$y_i$ path contained within $\eta$ is an identity path, it follows that the closed walks $W_1,\dots, W_k$ are all non-null closed walks. Furthermore, it is easy to see that every vertex of $G$ appears in at most 2 of these closed walks. This is because these walks were constructed by simply taking the sum of two linkages. Let $C_1,\dots, C_k$ denote non-null cycles in $G$ such that $C_i$ is contained in $W_i$. Since each $W_i$ is non-null, such cycles exist, implying the presence of a $\frac{1}{2}$-integral $k$-packing of non-null cycles. It remains to address the second case. That is, there is a set $X$  of at most $2k$ vertices which intersects every non-null $S$-path in $G$.

	We now use Observation \ref{obs:unique_block} to infer that there is a unique block $\cB$ in $G-X$ which intersects each $\eta(v_i)$ disjoint from $X$. We now observe that this block is clean. If this is not the case, then there are distinct indices $i,j\in [\ell]$ and vertices $u\in \eta(v_i)$ and $v\in \eta(v_j)$ such that $u,v\in \cB$, $\eta(v_i)$ and $\eta(v_j)$ are disjoint from $X$ and there is a non-null cycle $C$ in $G-X$ that contains both $u$ and $v$ (Observation \ref{obs:non-null-cycle-block}). Now, let $Q_u$ denote the unique path in $\eta(v_i)$ from $\gamma(v_i)$ to $u$. Similarly, let $Q_v$ denote the unique path in $\eta(v_j)$ from $\gamma(v_j)$ to $v$. By definition, $Q_u$ and $Q_v$ are vertex disjoint. Let $\alpha$ denote the \emph{first} vertex of the cycle $C$ encountered when traversing  $Q_u$ from $\gamma(v_i)$ to $u$. Similarly, let $\beta$ denote the \emph{first} vertex of the cycle $C$ encountered when traversing  $Q_v$ from $\gamma(v_j)$ to $v$. Since $C$ is non-null, it contains a pair of $\alpha$-$\beta$ paths $J_1$ and $J_2$ such that $\Lambda(J_1)\neq \Lambda(J_2)$. Hence, we infer the presence of  a pair of $\gamma(v_i)$-$\gamma(v_j)$ paths with differing values in $G-X$. This implies that $G-X$ contains a non-null $S$-path, a contradiction. Hence, we conclude that the block $\cB$ is clean.
	
	Let $z_1,\dots, z_r$ denote the cut-vertices of $G-X$ which are in the block $\cB$. Let $G_i$ denote the graph induced on $z_i$ and those vertices which have a path to $z_i$ without intersecting any other vertex of $\cB$. Observe that the graphs $G_1,\dots, G_r$ are vertex-disjoint. 
	If there are at least $k$ cut-vertices say $z_1,\dots, z_k$ such that  the graphs $G_1,\dots, G_k$ are not clean, then we are done since we have $k$ vertex disjoint non-null cycles. Otherwise, there are at most $k-1$ vertices $z_1,\dots, z_q$ (where $q\leq k-1$) such that the corresponding graphs $G_1,\dots, G_q$ are not clean. We now define $X'=X\cup \{z_1,\dots, z_q\}$. Since $q\leq k-1$, it follows that $|X'|\leq 3k$.
	
	Since $\ell>6k^2$, it follows that $\ell'>6k$ and hence we may conclude that there is an $i\in [\ell']$ such that $|X'|$ is disjoint from $\eta(v_i)$. Let $C$ denote the component of $G-X'$ which contains $\eta(v_i)$. We now argue that $G[C]$ is clean. If this were not the case, then there is a non-null cycle in $G[C]$. Furthermore, recall that the block $\cB$ in $G-X$ contains at least one vertex from every supernode of $\eta$. Hence, this non-null cycle is present in the connected component of $G-X$ which contains the block $\cB$. However, we have already concluded that $\cB$ is clean. Hence, this non-null cycle is contained in the graph $G_s$ for some $s\in [q]$. But by the definition of $X'$, we know that the vertex $z_s$ is contained in $X'$, implying that there is no path in $G-X'$ from $\eta(v_i)$ to this particular non-null cycle, a contradiction. Hence we conclude that $G[C]$ is clean.

	We now define the separation $(A,B)$ as $A=C\cup X'$ and $B=X'\cup (V(G)\setminus C)$. Since $G[A\setminus B]=G[C]$ is clean by definition, the first property of the separation is satisfied.
	Furthermore, observe that $A\cap B=X'$ which has size at most $3k$. Also, if $|A\cap B|=1$, then it follows that there is an arc with both endpoints in $A\setminus B$ which is \emph{not} part of a non-null cycle in $G$, which is a contradiction to our assumption that in every graph we consider, each arc is part of a non-null cycle.
	This implies that $|A\cap B|>1$ and hence the second property of the separation is satisfied. 
	 It only remains to argue that the final property is satisfied as well. That is, every supernode in $\eta^\star$ disjoint from $X'$ is contained in $A\setminus B$. But observe that every supernode in $\eta$ which is disjoint from $X'$ is indeed contained in $A\setminus B$ because at least one supernode of $\eta$ is contained entirely in $A\setminus B$. Since any supernode of $\eta$ is also a supernode of $\eta^\star$, the required property is satisfied. Finally, the time required to compute this separation is dominated by the time required to execute the algorithm of Proposition \ref{prop:gallai}, which is a polynomial time algorithm. 
	This completes the proof of the lemma.
\end{proof}

Observe that the only significant difference between the separation in the premise of Lemma \ref{lem:fine_linkage1} and the separation returned by the algorithm of Lemma \ref{lem:fine_linkage2} is that in the former, the graph $G[A]$ is clean and in the latter the graph $G[A\setminus B]$ is clean. Hence, we must still handle the case when 
 the algorithm of Lemma \ref{lem:fine_linkage2} returns a separation where $G[A]$ is not clean. However, in this case, we can simply delete the set $A\cap B$ and recurse on the graph $G[B\setminus A]$ with the parameter $k-1$. The correctness of doing so is formally described in the proof of Theorem \ref{thm:erdos-posa}.
This completes the subsection dealing with graphs containing a large order clique-expansion.

\subsection{Graphs excluding a large clique-expansion}\label{subsec:flatwall}

%
%

In this subsection, we describe our algorithm for the case when the graph under consideration has large treewidth but does not have a clique-expansion whose size exceeds $\rho(k)$. In this case, we will use the recently improved bounds for the Flat Wall Theorem proved by Chuzhoy \cite{Chuzhoy16} to infer the existence of a small `apex set' whose deletion leaves a graph with a large enough flat wall. Working with this flat wall, we will show that either we can find the desired $\frac{1}{2}$-integral $k$-packing of non-null cycles or we can again find a small set of vertices, say $X$ and a connected component of $G-X$, say $C$ which is clean and  contains a large well-linked set. Finally, we will be able to argue that we can either greedily add $X$ to the gfvs we are constructing or we can invoke Lemma \ref{lem:tw_reduction_one} on the graph $G[C\cup X]$ to find a vertex whose deletion from $G$ results in an equivalent graph. We begin by recalling basic definitions relating to walls in graphs with high treewidth.

\paragraph{Walls.} An \emph {elementary wall} of height $h$ is the graph depicted in Figure \ref{fig:wall} (see also \cite{Chuzhoy16}).
The rows are numbered $R_1,\dots, R_{h+1}$ and the columns are numbered $C_1,\dots, C_{h+1}$. The {\em bricks} of the wall are the 6 cycles created by $R_i,R_{i+1},C_{j},C_{j+1}$ for $1\leq i,j<h$ (see Figure \ref{fig:wall}). Each elementary wall of height $h$ consists of $h$ levels, each containing $h$ bricks. A {\em wall} of height $h$ is obtained from an elementary wall of height $h$ by subdividing some of the edges. We are now ready to state and prove a lemma that is analogous to Lemma \ref{lem:fine_linkage1}. Recall that  this lemma guarantees the existence of a function $\rho:{\mathbb N}\to {\mathbb N}$ which satisfies certain properties. We will use this function in the premise of the next lemma as well as in that of subsequent lemmas.

\begin{lemma}\label{lem:wall_linkage1}
	There is a function $\pi(x)=2^{2^{\Theta(x\log x)}}$ such that 
	for every $k\in \mathbb N$ and $\Gamma$-labeled graph $G$, if $G$ has a $W_{\pi(k)}$-wall $\cW$ and a separation $(A,B)$  such that \begin{enumerate}
	\item $G[A]$ is clean,  
	\item $1<|A\cap B|\leq 3k+\rho(k)$ and
	\item $\cW$ is contained in $A\setminus B$
 \end{enumerate}

then there is a vertex $v$ in $A\setminus B$ such that 
the graph $G-v$ is $k$-equivalent to the graph $G$ and given $G,k, \cW$ and the separation $(A,B)$, the vertex $v$ can be computed in time $2^{2^{\Theta(k\log k)}}n^{\bigoh(1)}$.

\end{lemma}

\begin{figure}[t]
\begin{center}

  \includegraphics{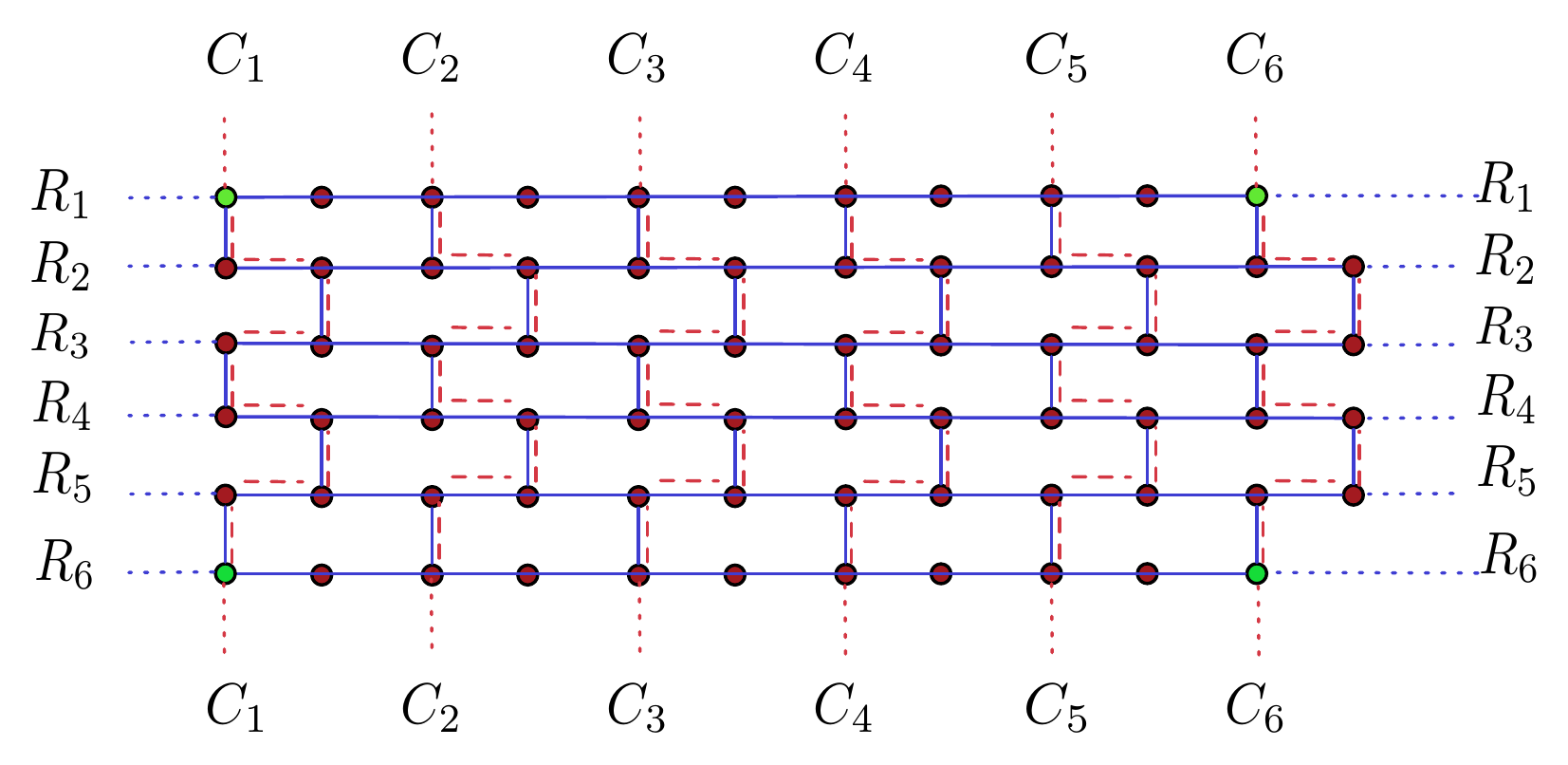}
  \end{center}
  \caption{An illustration of an elementary wall of height 5.}
  \label{fig:wall}
\end{figure}

\begin{proof}
Let $r=3k+\rho(k)$ and $\pi(k)=2(2^r \cdot r^{6r})+1)$. This implies the presence of a   well-linked set $Z$ of size $>2^r\cdot r^{6r}$ which is contained in $A\setminus B$. We then invoke Lemma \ref{lem:linkage2} on the graph $G[A]$ with the terminal set $T=A\cap B$, $p=|A\cap B|\leq r$ and the well-linked set $Z$ to compute a vertex $v\in Z$ such that $G-v$ is $k$-equivalent to $G$. The claimed running time follows from that of Lemma \ref{lem:linkage2} and the definition of $r$ and the function $\rho$. This completes the proof of the lemma.	
\end{proof}

What remains is to prove a lemma similar to Lemma \ref{lem:fine_linkage2}.  That is, a lemma which guarantees  a separation that (almost) satisfies the conditions in the premise of the above lemma.  Before we do so, we recall the requisite definitions regarding \emph{flat walls}.

\paragraph{Flat walls.} 
The {\em perimeter} of a wall are the vertices which appear on the subdivided paths corresponding to the edges in rows $R_1,R_{h+1}$ and the columns $C_1,C_{h+1}$. We refer to the vertices in the perimeter of the wall $\cW$ as $\per(\cW)$. The {\em corners} of a wall are the 4 vertices $R_1\cap C_1$, $R_1\cap C_{h+1}$, $R_{h+1}\cap C_1$ and $R_{h+1}\cap C_{h+1}$.
The \emph{nails} of a wall are the vertices of degree 3 within it and its corners. The \emph{pegs} of a wall are the vertices of degree 2 on the outer boundary.
For any wall $\cW$ in a graph $G$, there is a unique component $U$ of $G-\per(\cW)$ containing $W-\per(\cW)$. The \emph{compass} of $\cW$, denoted by $\comp(\cW)$, consists of the graph with vertex set $V(U)\cup \per(\cW)$ and edge set $E(U)\cup E(\per(\cW))\cup \{xy|x\in V(U), y\in V(\per(\cW))\}$. A {\em subwall} of a wall $\cW$ is a wall which is a subgraph of $\cW$. A subwall of $\cW$ of height $\ell$ is  proper if it consists of $\ell$ consecutive bricks  from each of $\ell$ consecutive levels of $\cW$. The {\em exterior} of a subwall $\cW'$ is $\cW-\cW'$. A proper subwall is {\em dividing} if its compass is vertex-disjoint from its exterior. 
We say that two dividing subwalls $\cW_1$ and $\cW_2$ of a wall $\cW$ are \emph{fully vertex-disjoint} if $\comp(\cW_1)$ is vertex-disjoint from $\comp(\cW_2)$. We now come to the definition of a \emph{flat wall}.

\begin{definition}{\sc \cite{Reed99}}
A wall $\cW$ is {\bf flat} if the graph $\comp(\cW)$ does not contain two vertex-disjoint paths connecting the diagonally opposite corners.
\end{definition}

\begin{proposition}
Any proper subwall of a flat wall is both flat and dividing. Furthermore, if $\cW_1$ and $\cW_2$ are proper subwalls of a flat wall $\cW$ and their perimeters are disjoint, then so are their compasses.	
\end{proposition}

\begin{definition}
	A wall $\cW$ in a $\Gamma$-labeled graph $G$ is called a {\bf null wall} if the subgraph $\comp(\cW)$ has no non-null cycles and a {\bf non-null wall} otherwise.
\end{definition}

Before we go ahead to the statement of Lemma \ref{lem:wall_linkage2}, we state the following results of Chuzhoy \cite{Chuzhoy15,Chuzhoy16} which will be used in the proof of this lemma and Theorem \ref{thm:erdos-posa}.

\begin{proposition}\label{prop:flat wall} {\sc [Flat Wall Theorem \cite{Chuzhoy15}]}
For all integers $\ell, t > 1$, every graph $G$ containing a wall of size $\Theta(t(t+\ell))$, must contain either
\begin{itemize}
\item a $K_t$-expansion or
\item a subset $A\subset V(G)$ of at most $t-5$ vertices and a flat wall of height $\ell$ in $G-A$.	
\end{itemize}

\noindent
Furthermore, there is an algorithm that runs in time polynomial in $|V(G)|,t,\ell$ and outputs either a $K_t$-expansion or a  subset $A\subset V(G)$ of at most $t-5$ vertices and a flat wall of height $\ell$ in $G-A$.

\end{proposition}

\begin{proposition}{\sc \cite{Chuzhoy16}}
	There is a function $f:{\mathbb N}\to {\mathbb N}$ such that $f(g)=\bigoh(g^{19}\log^{\bigoh(1)}g)$ and every graph of treewidth at least $f(g)$ has a $g \times g$-grid minor and a wall of size $\Omega(g)$.
\end{proposition}

We now proceed to prove an analogue of Lemma \ref{lem:fine_linkage2}.

\begin{lemma}\label{lem:wall_linkage2}
	There is a function $\sigma(x)=2^{2^{\Theta(x\log x)}}$ 
	such that 
	for every $k\in \mathbb N$, and $\Gamma$-labeled graph $G$, if $G$ excludes a $K_{\rho(k)}$-expansion and has a wall of height $\Theta(\sigma(k)(\sigma(k)+\rho(k)))$,   then it has a $\frac{1}{2}$-integral $k$-packing of non-null cycles or 
	 a separation $(A,B)$  such that \begin{enumerate}
	\item $G[A\setminus B]$ is clean,  
	\item $1<|A\cap B|\leq \rho(k)+3k$ and
	\item there is a $W_{\pi(k)}$-wall $\cW^\star$ which is contained in $A\setminus B$ (where $\pi$ is the function from Lemma \ref{lem:wall_linkage1}).
 \end{enumerate}
 
 \noindent
 Furthermore, there is an algorithm that, given $G,k$, runs in time $2^{2^{\Theta(k \log k)}}n^{\bigoh(1)}$ and returns either  a $\frac{1}{2}$-integral $k$-packing of non-null cycles or the separation $(A,B)$.

\end{lemma}

\begin{figure}
\centering
\begin{minipage}{.5\textwidth}
  \centering
  \includegraphics[width=.9\linewidth]{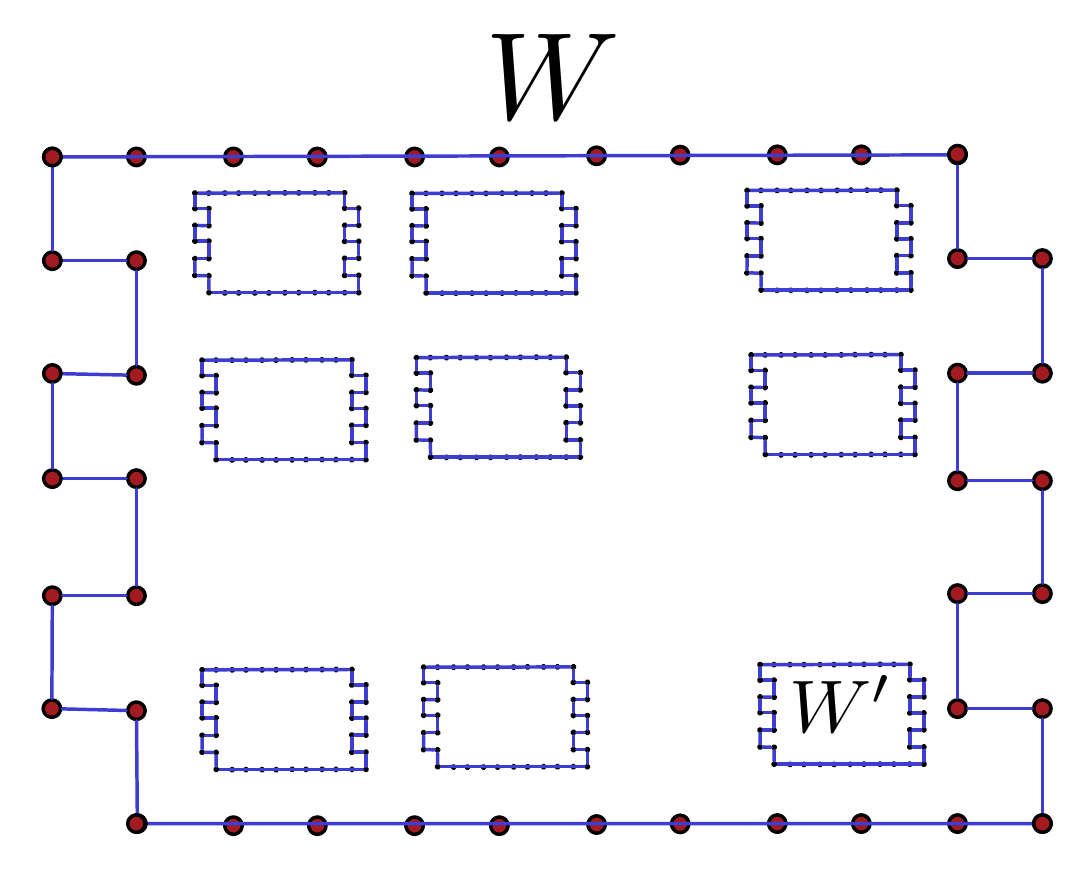}
\end{minipage}%
\begin{minipage}{.5\textwidth}
  \centering
  \includegraphics[width=.9\linewidth]{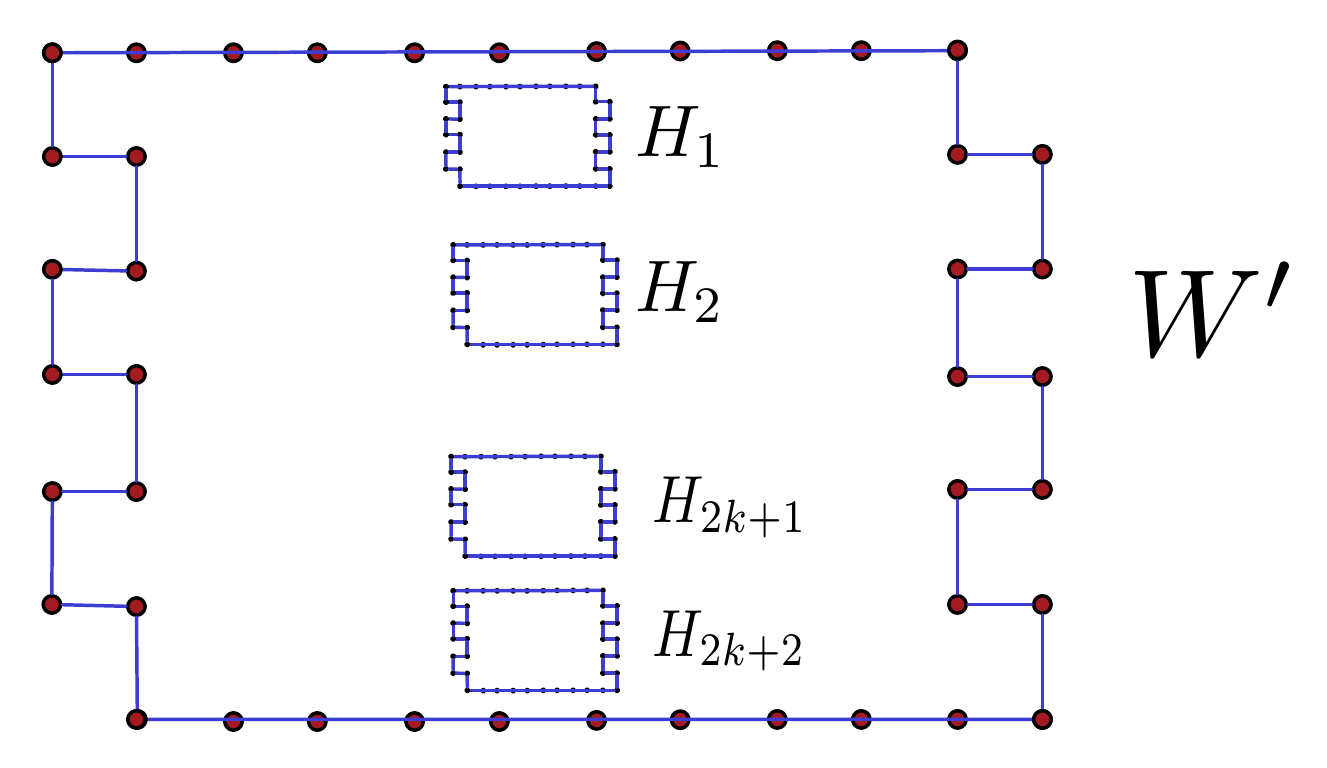}

\end{minipage}
\caption{Illustrations of the flat wall $\cW$, the flat, dividing and null subwall $\cW'$ and the fully disjoint subwalls $H_1,\dots, H_{2k+2}$ contained in $\cW'$. }
  \label{fig:manywalls}
    \label{fig:littlewalls}

\end{figure}

%
%

\begin{proof} We set $\sigma(k)=16k^2(\pi(k)+2k)$. By the Flat Wall Theorem (Proposition \ref{prop:flat wall}), we infer that $G$ contains a vertex set $\bA$ of size at most $\rho(k)$ such that $G-\bA$ has a flat wall $\cW$ of height $\sigma(k)$.

We now divide $\cW$ into $k$ fully vertex-disjoint flat and dividing proper subwalls  (see Figure \ref{fig:littlewalls}) each of height $8k(\pi(k)+2k)$. If each of these proper subwalls are non-null, then we are done since we have an \emph{integral} $k$-packing of non-null cycles. Hence, we may conclude that there is a null wall $\cW'$ of height $8k(\pi(k)+2k)$. Recall that this implies that the graph $\comp(\cW')$ is clean. Hence, we may assume without loss of generality that any arc of $G$ with both endpoints in the vertex set of $\comp(\cW')$ is labeled with the identity element $1_\Gamma$.

We now divide $\cW'$ into $4k$ pairwise fully vertex-disjoint flat and dividing proper subwalls $H_1,\dots, H_{4k}$ each of height $s=\pi(k)+2k$ such that the first and last columns of each of these walls are `subcolumns' of the same column of $\cW'$ (see Figure \ref{fig:manywalls}). 
 We now consider the wall $H_k$ and let $R_1,\dots ,R_{s+1}$ and $C_1,\dots, C_{s+1}$ denote the rows and columns that make up this wall. 
   We now pick $6k$ pegs on the \emph{bottom row} of $H_k$ which we call Row $R_1$, avoiding the corners. For each peg, we fix a unique vertex adjacent to it. We refer to these selected vertices as $p_1,\dots, p_{6k}$ and without loss of generality, let $p_i$ denote the vertex in the intersection of $R_1$ and $C_{1+i}$. Let $T=\{p_1,\dots, p_{6k}\}$.
   
   We now execute the algorithm of Proposition \ref{prop:gallai} in the graph $G'=G-\bA - (V(\comp(H_k))\setminus T)$ with the terminal set $T$. That is, we look for non-null $T$-paths in the graph obtained from $G$ by removing the apex set $\bA$ \emph{and} the interior of the wall $H_k$ and the vertices on the perimeter of $H_k$ \emph{except} those vertices in $T$. We do this so that if we find the required set of non-null paths, then we can use the interior of the wall $H_k$ to complete these into a set of $k$ non-null cycles.

  Suppose that the algorithm of Proposition \ref{prop:gallai} returned a set of $k$-vertex-disjoint non-null $T$-paths, $P_1,\dots, P_k$ and let $p_{i_1},p_{i_2}$ denote the vertices of $T$ on the endpoints of path $P_i$. 
  We now define the cycles $C_1,\dots, C_k$ where $C_i$ is defined as the cycle obtained by taking the sum of the path $P_i$, row $R_{1+i}$ in the wall $H_1$ and the two columns $C_{i_1+1}$ and $C_{i_2+1}$. Since $H_k$ is a null-wall and $P_i$ is a non-null path, it follows that each $C_i$ defined in this way is a non-null cycle. Furthermore, any vertex of $G'$ appears in at most one of these cycles since the paths $P_1,\dots, P_k$ are vertex-disjoint. On the other hand, any vertex in $V(G)-V(G')$ can be easily seen to appear in at most 2 of these non-null cycles, implying a $\frac{1}{2}$-integral $k$-packing of non-null cycles.

         \begin{figure}[t]
\begin{center}

  \includegraphics{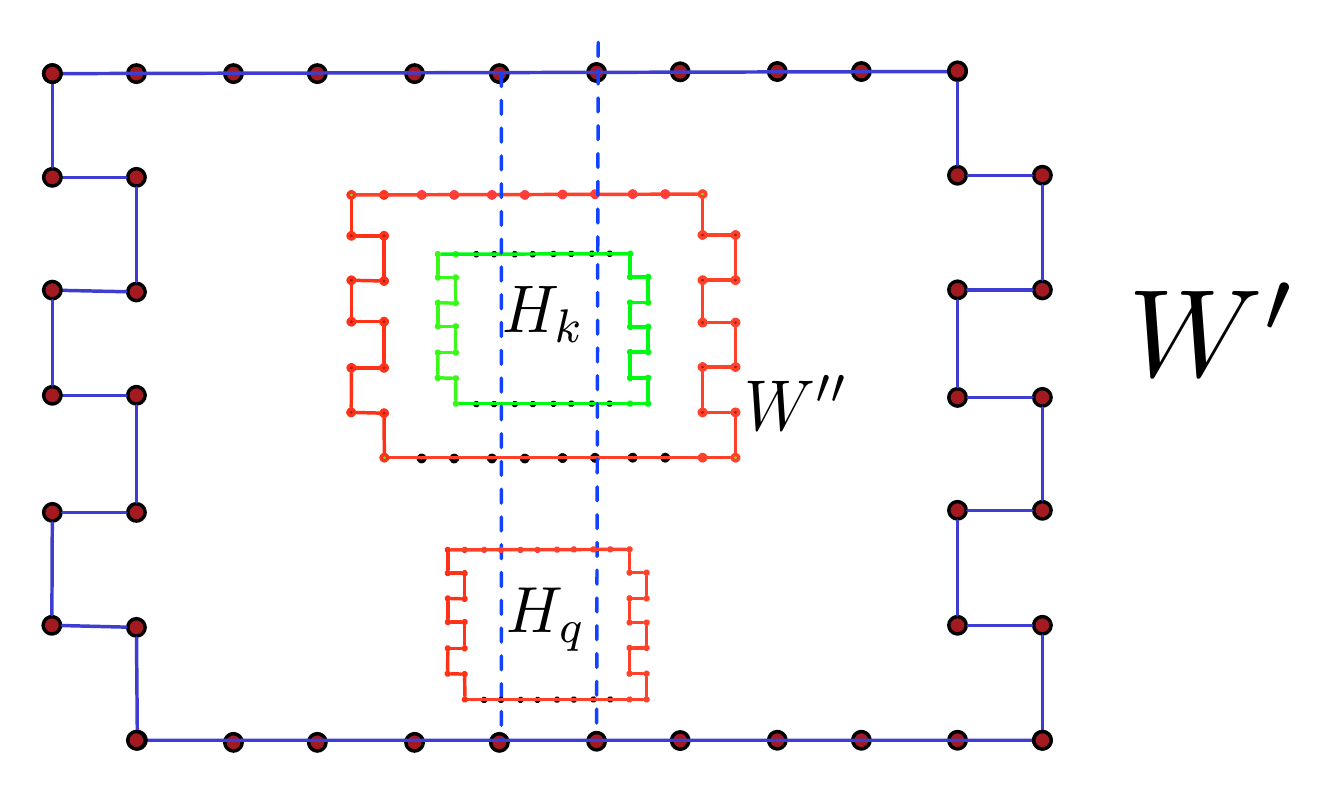}
  \end{center}
  \caption{An illustration of the walls $W''$ and $H_q$. The two columns of $W'$ in the figure intersect the subwall $H_k$ at the vertices $p_i$ and $p_j$.}
  \label{fig:doublewall}
\end{figure}
  
     Otherwise, suppose that the algorithm of Proposition \ref{prop:gallai} returns a set $X$ of size at most $2k$ which intersects all non-null $T$-paths in $G'$. Since $|X|\leq 2k$, it follows that out of the $6k$ vertices in $T$, there are $4k$ vertices which are not in $X$ and furthermore, among these $4k$ vertices there are two vertices $p_i$ and $p_j$ such that $X$ is disjoint from the columns of $\cW'$ which contain the columns $C_{1+i}$ and $C_{1+j}$ of $H_k$. Similarly, since $|X|\leq 2k$, it follows that there is a proper subwall of $\cW'$, call it $\cW''$ such that $H_k$ itself is a proper subwall of $\cW''$ and $X$ is disjoint from $\per(\cW'')$. That is, $\cW''$ is `nested' between $H_k$ and $\cW'$ and its perimeter is disjoint from $X$. Similarly, there is a wall $H_q$ such that $q\in \{k+2,\dots, 4k\}$ and $X$ is disjoint from $V(\comp(H_q))$. Finally, observe that we can select $\cW''$ and $H_q$ in such a way that $\cW''$ and $H_q$ are fully vertex-disjoint (see Figure \ref{fig:doublewall}).
     
     Observe that since $\per(H_q)$ is disjoint from $X$, it follows that there is a block of $G'-X$ that contains this set of vertices.
  We now argue that this block is clean. Suppose not. Then, there is a non-null cycle in $G'-X$ containing any pair of vertices in $\per(H_q)$ (Observation \ref{obs:non-null-cycle-block}). But $G'-X$ also has 2 vertex-disjoint null paths from $p_i$ and $p_j$ to the set $\per(H_q)$ (simply use the columns of $\cW'$ which contain the columns $C_{1+i}$ and $C_{1+j}$ in $H_k$) and hence to this non-null cycle as well. But this implies the presence of a non-null $T$-path in $G'$ disjoint from $X$, a contradiction. Hence, we conclude that the block of $G'-X$ containing the vertices of $\per(H_q)$, call it $\cB$, is clean.
  We now argue that adding back the vertices of $\comp(H_k)$ that we originally removed will not create any new non-null $\per(H_q)$-paths.
  
  \begin{claim} The block of $G''=G-(X\cup \bA)$ that contains $\per(H_q)$, call it $\cB'$, is also clean. 
  	
  \end{claim}
  
  \begin{proof}
Suppose not. Then, there is a non-null path between every pair of vertices in $\per(H_q)$ in the graph $G''$. Clearly, any such non-null path must intersect the vertices in $\comp(H_k)$. However, since $H_k$ is a proper subwall of $\cW''$ which is fully vertex-disjoint from $H_q$, it follows that any such non-null path intersects $\per(\cW'')$ in at least 2 vertices. Let $P$ be such a path and let $J_1,\dots, J_s$ be the subpaths of $P$ whose endpoints are on $\per(\cW'')$ and whose internal vertices are in $\comp(\cW'')\setminus \per(\cW'')$. Since $\cW''$ is a null-wall, it follows that these subpaths are all identity paths. We now construct a non-null walk $P'$ in $G'-X$ as follows. For every path $J_x$, we replace it with the identity path between the endpoints of $J_x$ which is contained in the graph induced on $\per(\cW')$.

     Clearly $P'$ is now a non-null walk between a pair of vertices in $\per(H_q)$ in the graph $G'-X$. However, this alone does not give us a contradiction. For that, we will show that $P'$ is in fact contained entirely in the block $\cB$, which will then contradict our earlier observation that $\cB$ is clean. 
     It is straightforward to see that all vertices of $\per(\cW'')$ are contained in the block $\cB$ in $G'-X$. Hence, we only need to argue about the vertices of $P'$ which are not in $\comp(\cW'')$. That is, those vertices of $P$ which are also present in $P'$. But these vertices lie on a path in $G'-X$ between a vertex of $\per(H_q)$ and a vertex of $\per(\cW'')$ or between a pair of vertices of $\per(\cW'')$ or between a pair of vertices of $\per(H_q)$. In each case, it follows that these vertices are present in the block $\cB$ in the graph $G'-X$, implying the presence of a non-null path in the clean block, a contradiction.
      Hence, we conclude that the block $\cB'$ in the graph $G''$ is clean. This completes the proof of the claim.
     \end{proof}

     Consider the graph $G''$ and the cut-vertices of the block $\cB'$. If there are $k$ distinct cut-vertices which either appear in a non-null cycle or separate the vertices in $\cB$ from a non-null cycle, then we are done since we have $k$ vertex-disjoint non-null cycles in $G''$ and hence in $G$. Therefore, we may assume that there are at most $k-1$ such cut-vertices. We add these vertices to the set $X\cup \bA$ to obtain the set $X'$. Clearly, the set $X'$ has size at most $\rho(k)+3k$ in total and  is disjoint from at least one row and column of the wall $H_q$.
      Let $C$ denote the component of $G-X'$ which contains this row and column. We claim that the graph induced on $C$ is clean. The argument is similar to that used in Lemma \ref{lem:fine_linkage2} and hence we do not repeat it.

     We now define the separation $(A,B)$ as follows. The set $A=C\cup X'$, $B=V(G)\setminus C$. We now argue that this separation satisfies the required properties. Since $A\setminus B=C$, it follows that $G[A\setminus B]$ is clean. Furthermore, since $A\cap B=X'$ and we know that $|X'|\leq \rho(k)+3k$ and since every arc in $G$ is part of a non-null cycle, it follows that $|A\cap B|>1$. Hence, the second property is also satisfied. Finally, since $H_q$ is present in $G-(\bA\cup X)$ and has size at least $\pi(k)+k$, it follows that there is a subwall of $H_q$ of size at least $\pi(k)$ which is contained entirely in $A\setminus B$ (playing the role of the wall $\cW^\star$), satisfying the third property.  The time required to compute this separation is essentially dominated by the time required to compute the set $\bA$ and the flat wall $\cW$. That is, the time required by the algorithm of Proposition \ref{prop:flat wall}, which is polynomial in $|V(G)|,\rho(k)$ and $\sigma(k)$.    
     This completes the proof of the lemma.
\end{proof}

     We complete this section by giving the full proof of Theorem \ref{thm:erdos-posa}. We restate it here for the sake of completeness.

     \maintheorem*

     \begin{proof} 
     
     Let $\rho'(k)=\rho(k)+3k$ and let $\sigma'(k)=\sigma(k)+\rho'(k)+3k$. Let $c$ be a constant derived from the constant hidden in the $\Theta$ notation of Proposition \ref{prop:flat wall} and scaled up sufficiently so that a graph with treewidth at least $c\cdot ((\sigma'(k)(\sigma'(k)+\rho'(k)))^{20})$ will have either a $K_{\rho'(k)}$-expansion or a vertex set $\bA$ of size at most $\rho'(k)$ and a flat wall of size at least $\sigma'(k)$ in $G-\bA$.

     We first consider the case when the graph has treewidth $<w$, where  $w=c\cdot ((\sigma'(k)(\sigma'(k)+\rho'(k)))^{20})$. 
     In this case, we invoke Lemma \ref{lem:bounded_treewidth_case} and conclude that either the graph has $k$-vertex-disjoint non-null cycles or a gfvs of size at most $(k-1)(w+1)$. 
So, we choose $\tau$ in a way that it satisfies the following inequalities. 
     \begin{center}
     \begin{itemize}
    
     	\item $\tau(k)\geq \tau(k-1)+\rho'(k)+3k$
     	 	\item $\tau(k)\geq (k-1)(w+1)$
     \end{itemize}
     \end{center}
     Since $w=2^{2^{\Theta(k \log k)}}$ and $\rho(k)=2^{\Theta(k \log k)}$, we can choose $\tau$ such that $\tau(k)=2^{2^{\Theta(k \log k)}}$.
     If the graph has treewidth $\geq w$ then we use the algorithm of Proposition \ref{prop:flat wall} to either compute a $k_{\rho'(k)}$-expansion $\eta$ or a vertex set $\bA$ of size at most $\rho'(k)$ and a flat wall $\cW$ of size $\sigma'(k)$ in $G-A$.

If the graph is found to have a $K_{\rho'(k)}$-expansion $\eta$, then we invoke Lemma \ref{lem:fine_linkage2} to either compute a $\frac{1}{2}$-integral $k$-packing of non-null cycles or a separation $(A,B)$ such that $G[A\setminus B]$ is clean, $1<|A\cap B|\leq 3k$, and $G[A\setminus B]$ contains a $K_{\rho(k)}$-expansion. 
If the graph $G[A]$ is not clean, then we recursively run the algorithm on the graph $G[B\setminus A]$ with parameter $k-1$. If the recursive call returns a $\frac{1}{2}$-integral $(k-1)$-packing then we add to this set an arbitrary non-null cycle in $G[A]$ to get the required $\frac{1}{2}$-integral $k$-packing of non-null cycles. On the other hand, if the recursive call returns a gfvs of $G[B\setminus A]$ of size at most $\tau(k-1)$, then this set along with the set $A\cap B$ is clearly a gfvs of size at most $\tau(k-1)+3k\leq \tau(k)$. Hence, we may assume that the graph $G[A]$ is clean. We then invoke Lemma \ref{lem:fine_linkage1} to compute a vertex $v$, delete it from the graph and recurse on the resulting graph which is known to be $k$-equivalent to the original graph.

In the second case, we invoke Lemma \ref{lem:wall_linkage2} to conclude that $G$ either has a $\frac{1}{2}$-integral $k$-packing of non-null cycles or a separation $(A,B)$ such that $G[A\setminus B]$ is clean, $|A\cap B|\leq \rho'(k)+3k$ and $G[A\setminus B]$ contains a $W_{\sigma(k)}$ wall.
If $G[A]$ is not-clean, then we do the same as before. That is, we recurse on the graph $G[B\setminus A]$ with parameter $k-1$ and use the obtained solution to either compute a $\frac{1}{2}$-integral $k$-packing or a gfvs of $G$ of size at most $\tau(k-1)+\rho'(k)+3k\leq \tau(k)$.
   Otherwise, we invoke Lemma \ref{lem:wall_linkage1} to compute a  vertex $v$, delete it and  recurse on the resulting graph which is $k$-equivalent to the input graph. The claimed running time in the theorem follows from those of Proposition \ref{prop:flat wall}, Lemma \ref{lem:fine_linkage1}, Lemma \ref{lem:fine_linkage2}, Lemma \ref{lem:wall_linkage1} and Lemma \ref{lem:wall_linkage2}.
   This completes the proof of the theorem.     	
     \end{proof}

\section{Conclusions}
We remark that in the case of planar graphs, the proof of Lemma \ref{lem:wall_linkage2} can be easily modified to produce an \emph{integral} $k$-packing as an output instead of a $\frac{1}{2}$-integral $k$-packing. Furthermore, since we can completely avoid the case of the graph having a large order clique-expansion and work directly with the case when the input graph has a large flat wall (without even the need to delete an apex set), the function $\tau$ can be chosen such that $\tau(k)=2^{\Theta(k \log k)}$.  As a result, our proof also implies  that non-null cycles on group-labeled \emph{planar graphs} have the {\erdosposa} property with the function $f(k,\cF)$ being bounded by a single-exponential function of $k$.
We leave open the existence of a \emph{polynomial} bound on the $\frac{1}{2}$-integral {\erdosposa} function for non-null cycles in general graphs.

\bibliographystyle{siam}
 \bibliography{references,references_new}

\end{document}